\documentclass[11pt]{article}

\newif\ifFULL
\FULLtrue

\usepackage{thmtools,thm-restate} 
\usepackage[margin=1in]{geometry} 
\usepackage{amsmath,amsthm,amssymb,amsfonts}
\usepackage{xcolor}
\usepackage{algorithm}
\usepackage{algorithmic}

\newcommand{\E}{\mathbb{E}}
\newcommand{\one}{\mathbf{1}}
\newcommand{\M}{\mathcal{M}}
\newcommand{\OPT}{\textsc{OPT}}

\usepackage[colorlinks,citecolor=blue]{hyperref}

\newcommand\Thomas[1]{\textcolor{black}{#1}}
\newcommand\Sahil[1]{\textcolor{black}{#1}}

\newtheorem{theorem}{Theorem}

\newtheorem{lemma}[theorem]{Lemma}

\newtheorem{definition}[theorem]{Definition}

\newtheorem{claim}[theorem]{Claim}

\newtheorem{example}[theorem]{Example}

\newcounter{note}[section]
\newcommand{\snote}[1]{\refstepcounter{note}$\ll${\bf Sahil~\thenote:}
	{\sf \color{red}  #1}$\gg$\marginpar{\tiny\bf SS~\thenote}}
\newcommand{\tnote}[1]{\refstepcounter{note}$\ll${\bf Thomas~\thenote:}
	{\sf \color{teal}  #1}$\gg$\marginpar{\tiny\bf TK~\thenote}}

\newcommand{\Rev}{\text{Revenue}}
\newcommand{\Util}{\text{Utility}}
\newcommand{\val}{v}
\newcommand{\bval}{\textbf{v}}
\newcommand{\baseprice}{b}
\newcommand{\arrivaltime}{{T}}
\newcommand{\barrivaltime}{\textbf{T}}
\newcommand{\IGNORE}[1]{}

\newcommand{\Alg}{\text{Alg}}

\newcommand{\growingmid}{\mathrel{}\middle|\mathrel{}}

\newcommand*\samethanks[1][\value{footnote}]{\footnotemark[#1]}

\title{
	Prophet Secretary for Combinatorial Auctions and Matroids\footnote{Part of this work was done while the authors were visiting the Simons Institute for the Theory of Computing.} 
	}

\author{Soheil Ehsani\thanks{
		Department of Computer Science, University of Maryland, College Park, MD 20742 USA. Email: \texttt{\{ehsani,hajiagha\}@cs.umd.edu}. Supported in part by NSF CAREER award CCF-1053605, NSF BIGDATA grant IIS-1546108, NSF AF:Medium grant CCF-1161365, DARPA GRAPHS/AFOSR grant FA9550-12-1-0423, and another DARPA SIMPLEX grant.
	}
\texttt{}	\and MohammadTaghi Hajiaghayi\samethanks[2]
	\and Thomas Kesselheim\thanks{Department of Computer Science, TU Dortmund, 44221 Dortmund, Germany. Email: \texttt{thomas.kesselheim@tu-dortmund.de}.
	}
	\and Sahil Singla\thanks{
		Computer Science Department,
		Carnegie Mellon University, Pittsburgh, PA 15213, USA. Email: \texttt{ssingla@cmu.edu}. Supported in part by a CMU Presidential Fellowship and NSF awards CCF-1319811, CCF-1536002, and CCF-1617790
	}
}

\date{ \today}

\begin{document}
\maketitle
\thispagestyle{empty}

\begin{abstract}{
The {secretary} and the {prophet inequality} problems are central to the field of Stopping Theory. Recently, there has been a lot of work in generalizing these models to multiple items because of their applications in mechanism design. The most important of these generalizations are to matroids and to combinatorial auctions (extends bipartite matching).  Kleinberg-Weinberg~\cite{KW-STOC12} and Feldman et al.~\cite{feldman2015combinatorial} show that for adversarial arrival order of random variables the optimal prophet inequalities give a $1/2$-approximation. For many settings, however, it's conceivable that the arrival order is chosen uniformly at random, akin to the secretary problem. For such a random arrival model, we improve upon the $1/2$-approximation and obtain $(1-1/e)$-approximation prophet inequalities for both matroids and combinatorial auctions. 
This also gives  improvements to the results of Yan~\cite{yan2011mechanism} and Esfandiari et al.~\cite{esfandiari2015prophet} who worked in the special cases where we can fully control the arrival order or when there is only a single item.

Our techniques are threshold based. We  convert our discrete problem into a continuous setting and then give a generic template on  how to dynamically adjust these thresholds to lower bound the expected total welfare.










}\end{abstract}

\clearpage
\setcounter{page}{1}


\section{Introduction}
Suppose there is a sequence of $n$ buyers arriving with different \emph{values} to your \emph{single} item. On arrival a buyer offers a take-it-or-leave-it value for your item. How should you decide which buyer to assign the item to in order to maximize the value. There are two popular models in the field of Stopping Theory to study this problem: the \emph{secretary} and the \emph{prophet inequality} models. In the secretary model we assume  no prior knowledge about the buyer values but the buyers arrive in a uniformly random order~\cite{dynkin1963optimum}. Meanwhile, in the prophet inequality model we assume stochastic knowledge about the buyer values but the arrival order of the buyers is chosen by an adversary~\cite{krengel1978semiamarts,krengel1977semiamarts}. Since the two models complement each other, both have been widely studied in the fields of mechanism design and combinatorial optimization (see related work). 

These models assume that either the buyer values or the buyer arrival order is chosen by an adversary. In practice, however, it is often conceivable that there is no adversary acting against you. Can we design better strategies in such settings?
The  \emph{prophet secretary} model introduced in \cite{esfandiari2015prophet} is a natural way to consider such a process where we assume both a {stochastic knowledge} about buyer values and that the buyers arrive in a uniformly random order. The goal is to design a strategy that maximizes expected accepted value, where the expectation is over the random arrival order, the stochastic buyer values, and also any internal randomness of the strategy. 


In this paper, we consider generalizations of the above problem to combinatorial settings. Suppose the buyers correspond to elements of a matroid\footnote{A matroid $\mathcal{M}$ consists of a ground set $[n]=\{1,2,\ldots,n\}$ and a non-empty downward-closed set system $\mathcal{I} \subseteq 2^{[n]}$ satisfying the matroid exchange axiom: for all pairs of sets $I, J\in \mathcal{I}$ such that $\lvert I \rvert < \lvert J \rvert$, there exists an element $x\in J$ such that $I \cup \{x\} \in \mathcal{I}$. Elements of $\mathcal{I}$ are called independent sets.} and we are allowed to accept any independent set in this matroid rather than only a single buyer. The buyers again arrive and offer take-it-or-leave-it value for being accepted.
In the prophet inequality model, a surprising result of Kleinberg-Weinberg~\cite{KW-STOC12} gives a $1/2$-approximation strategy to this problem, i.e., the value of their strategy, in expectation, is at least {half} of the value of the expected offline optimum that selects the best set of buyers in \emph{hindsight}. Simple examples show that for adversarial arrival one cannot improve this factor. On the other hand, if we are also allowed to  control the arrival order of the buyers, Yan~\cite{yan2011mechanism} gives a $1-1/e\approx 0.63$-approximation strategy. But what if the arrival order is neither adversarial and nor in your control. In particular, can we beat the $1/2$-approximation for a uniformly random  arrival order? 
 


\vspace{0.2cm}
\noindent \textbf{Matroid Prophet Secretary Problem (MPS):} \emph{Given a matroid $\M=([n],\cal{I})$ on $n$ buyers (elements) and  independent  probability distributions on their values, suppose the outcome buyer values are revealed in a uniformly random order. Whenever a buyer value is revealed,   the problem is to immediately and irrevocably decide whether to \emph{select} the buyer. The goal is to maximize the sum of values of the selected buyers, while  ensuring that they are always feasible in $\cal{I}$.}
\vspace{0.2cm}

Besides being a natural problem that relates two important Stopping Theory models, MPS is also interesting because of its applications in mechanism design. Often while designing mechanisms, we have to balance between maximizing revenue/welfare and the simplicity of the mechanism. While there exist optimal mechanisms such as VCG or Myerson's mechanism, they are impractical in real markets~\cite{AM-CA06,Rothkopf-07}. On the other hand, simple \emph{Sequentially Posted Pricing mechanisms}, where we offer take-it-or-leave-it prices to buyers, are known to give good approximations to optimal mechanisms. This is because the problem gets reduced   to designing a prophet inequality~\cite{chawla2009sequential,yan2011mechanism,alaei2014bayesian,KW-STOC12,feldman2015combinatorial}. 

Esfandiari et al.~\cite{esfandiari2015prophet} study MPS in the  special case of a  rank~$1$ matroid and give a $(1-1/e)$-approximation algorithm. For general matroids, as in the original models of~\cite{chawla2009sequential,yan2011mechanism,KW-STOC12}, it was unclear prior to the work of this paper whether beating the factor of $1/2$ is possible. In Section~\ref{sec:matroid} we prove the following result.


\begin{restatable}{theorem}{MPS}\label{thm:matProphetSec}
There exists a $(1-1/e)$-approximation algorithm to MPS.
\end{restatable} 
\noindent Note that the approximation in this theorem as well as the following ones compare to the expected \emph{optimal offline solution} for the particular outcomes of the distributions. That is, in the case of matroids, we have $\E[\Alg] \geq (1-1/e) \cdot \E[ \max_{I \in \mathcal{I}} \sum_{i \in I} v_i ]$, where $v_i$ is the  value of buyer $i${\footnote{\Sahil{It's not known if $1-1/e$ is tight for MPS. In fact, it's even open if one can beat $1-1/e$  for a single item~\cite{abolhassani2017beating}.}}.

Next, let us consider a  combinatorial auctions setting. Suppose there are $n$ buyers that take combinatorial valuations (say, submodular) for $m$ indivisible items from $n$ independent probability distributions. The problem is to decide how to allocate the items to the buyers, while trying to  maximize the \emph{welfare}---the sum of valuations of all the buyers.
Feldman et al.~\cite{feldman2015combinatorial} show that for  XOS\footnote{A function $v\colon 2^M\rightarrow \mathcal{R}$ is an XOS function if there exists a collection of additive functions $A_1,\ldots,A_k$ such that for every $S\subseteq M$ we have $v(S)=\max_{1\leq i\leq k} A_i(S)$.} (a generalization of submodular) valuations there exist \emph{static prices} for items that gets a $1/2$-approximation for buyers arriving in an adversarial order. Since this factor cannot be improved for adversarial arrival, this leaves an important open question if we can design better algorithms when the  arrival order  can be controlled. Or ideally, we want to  beat $1/2$  even when the arrival order cannot be controlled but is chosen uniformly at random.

\vspace{0.2cm}
\noindent \textbf{Combinatorial Auctions Prophet Secretary Problem (CAPS):} \emph{Suppose $n$ buyers take XOS valuations for $m$ items from $n$ independent probability distributions. 
The outcome buyer valuations are revealed in a uniformly random order. Whenever a buyer valuations is revealed, the problem is to immediately and irrevocably assign a subset of the remaining items to the buyer. The goal is maximize the sum of the valuations of all the buyers for their assigned subset of items.
}
\vspace{0.2cm}

In Section~\ref{sec:XOSCombAuctions} we improve the \Thomas{online approximation} result of~\cite{feldman2015combinatorial} for random order.

\begin{restatable}{theorem}{CAPS}\label{thm:XOS}
There exists a  $(1-1/e)$-approximation algorithm to CAPS.
\end{restatable}

\noindent Given access to demand and XOS oracles for stochastic utilities of different buyers, the algorithm in Theorem~\ref{thm:XOS} can be made efficient. This  is  interesting because it matches the best possible $(1-1/e)$-approximation for XOS-welfare maximization  in the offline setting~\cite{DNS-MOR10,Feige-SICOMP09}.

A desirable property in the design of an economically viable mechanism is \emph{incentive-compatibility}. In particular, a buyer is more likely to make decisions about their allocations based on their own personal incentives rather than to accept a given allocation that might optimize the social welfare but not the individuals' profit. For the important case of unit-demand buyers (aka bipartite matching), in Section~\ref{sec:matching} we extend Theorem~\ref{thm:XOS} to additionally obtain this property.


\begin{restatable}{theorem}{bipMatching} \label{thm:bipMatching}
For bipartite matchings, when buyers arrive in a uniformly random order, there exists an incentive-compatible mechanism \Thomas{based on dynamic prices} that gives a $(1-1/e)$-approximation to the optimal welfare.
\end{restatable}


\Thomas{For this result, we require unit-demand buyers. This is because for general XOS functions shifting buyers to earlier arrivals can change the availability of items arbitrarily. For unit-demand functions, we show that this effect is bounded.}

Finally, in Section~\ref{section:fixed} we conclude by showing that for the single-item case one can obtain a $(1-1/e)$-approximation even by using static prices, and that nothing better is possible.

\IGNORE{
\begin{theorem}
	There is no static-pricing mechanism for (single item) prophet secretary that achieves an approximation factor of $(1-1/e+\epsilon)$ for any constant $\epsilon>0$.
\end{theorem}
}


\subsection{Our Techniques}
In this section we discuss our three main ideas for a  combinatorial auction. In this setting, our algorithm is threshold based, which means that we set \emph{dynamic prices} to the items and  allow a buyer to purchase a set of items only if her value is more than the price of that set. This allows us to view  total value as the sum of \emph{utility} of the buyers and the total generated \emph{revenue}. Although powerful, dynamic prices  often lead to  involved calculations and become difficult to analyze beyond a  single item setting~\cite{esfandiari2015prophet,abolhassani2017beating}. 
To overcome this issue, we convert our discrete  problem into a continuous setting. \Thomas{This is possible because} a random permutation of buyers can be  viewed as each buyer arriving at a time chosen uniformly at random between $0$ and $1$. The benefit of such a transformation is that \Thomas{the arrival times are independent, which keeps correlations managable. Besides, it allow us to use tools  from integral calculus such as integration by parts.}

Our algorithm for combinatorial auctions sets a \emph{base price} $b_j$ for every item $j$ based on its contribution to the expected offline optimum $\E[\OPT]$. Our approach is to define two time varying continuous functions: \emph{discount} and \emph{residual}. The {discount} function $\alpha(t)\colon [0,1]\rightarrow [0,1]$ is chosen such that the price of an unsold item $j$ at time $t$ is exactly $\alpha(t) \cdot b_j$.  We define  a \emph{residual} function $r(t) \colon [0, 1] \to \mathbb{R}_{\geq 0}$ that intuitively denotes the expected value remaining in the instance at time $t$. Hence, $r(0)=\E[\OPT]$ and $r(1)=0$. Computing $r(t)$ is difficult   for a combinatorial auction since it depends on several random variables. However, assuming that we know $r(t)$, we use application specific techniques  to compute lower bounds on both the expected revenue and the expected utility in terms of the  functions  $r(t)$ and $\alpha(t)$. 

Finally, although we do not know  $r(t)$, 
we can choose the function $\alpha(t)$ in a way that allows us to simplify the sum of expected revenue and utility, without ever computing $r(t)$ explicitly. This step exploits  properties of the exponential function for integration (see Lemma~\ref{lem:residualSuffices}).

\subsection{Related Work}
Starting with the works of Krengel-Sucheston~\cite{krengel1978semiamarts,krengel1977semiamarts} and Dynkin~\cite{dynkin1963optimum}, there has been a long line of research on both prophet inequalities and secretary problems. 
One of the first generalizations  is the \textit{multiple-choice prophet inequalities} \cite{kennedy1987prophet,kennedy1985optimal,kertz1986comparison} in which we are allowed to pick $k$ items and the goal is to maximize their sum. Alaei~\cite{alaei2014bayesian} gives an almost tight ($1-{1}/{\sqrt{k+3}}$)-approximation algorithm for this problem (the lower bound is due to~\cite{hajiaghayi2007automated}). Similarly, the  \textit{multiple-choice secretary} problem   was first studied by  Hajiaghayi et al.~\cite{hajiaghayi2004adaptive},  and
 Kleinberg~\cite{kleinberg2005multiple} gives a $(1-O(\sqrt{1/k}))$-approximation algorithm.

The research investigating the relation between prophet inequalities and online auctions is
initiated in~\cite{hajiaghayi2007automated,chawla2009sequential}. This lead to several interesting follow up works for matroids~\cite{yan2011mechanism,KW-STOC12} and matchings~\cite{alaei2012online}. Meanwhile, the connection between secretary problems and online auctions is first explored in Hajiaghayi et al.~\cite{hajiaghayi2004adaptive}.
Its generalization  to matroids is considered in~\cite{babaioff2007matroids,Lachish-FOCS14,FSZ-SODA15} and to matchings in~\cite{GM-SODA08,KorulaPal-ICALP09,MY-STOC11,KMT-STOC11,kesselheim2013optimal,GS-IPCO17}.

Secretary problems and prophet inequalities have also been studied beyond a matroid/matching. For the intersection of $p$ matroids, Kleinberg and Weinberg~\cite{KW-STOC12} give an $O(p)$-approximation prophet inequality. D{\"u}tting and Kleinberg~\cite{dutting2015polymatroid} extend this result to polymatroids. Feldman et al.~\cite{feldman2015combinatorial} study the generalizations to combinatorial auctions. Later, D{\"u}tting et al.~\cite{DFKL-FOCS17} give a general framework to prove such prophet inequalities.
 Submodular variants of the secretary problem have been considered in~\cite{bateni2010submodular,gupta2010constrained,FZ-FOCS15,KMZ-STOC15}. Prophet  and secretary problems have also been studied for many classical combinatorial problems (see e.g., \cite{Meyerson-FOCS01,garg2008stochastic,gobel2014online,dehghani2015online,dehghani_et_al:LIPIcs:2017:7480}).
Rubinstein~\cite{rubinstein2016beyond} and Rubinstein-Singla~\cite{RS-SODA17} consider these problems for arbitrary downward-closed constraints.  

In the prophet secretary model, Esfandiari et al.~\cite{esfandiari2015prophet} give a $(1-1/e)$-approximation in the special case of a single item. Going beyond $1-1/e$ has been challenging. Only recently, Abolhasani et al.~\cite{abolhassani2017beating} and Correa et al.~\cite{CFHOV-EC17}  improve this  factor for the single item i.i.d. setting. 
Extending this result to non-identical  items or to matroids are interesting open problems.

\section{Our Approach using a Residual} \label{section:approach}

In this section, we define a residual  and discuss how it can be used to design an approximation algorithm for a prophet secretary problem. Suppose there are $n$ requests that arrive at times $(\arrivaltime_i)_{i \in [n]}$  drawn i.i.d. from the uniform distribution in $[0,1]$. These requests correspond to buyers of a combinatorial auction or to elements of a matroid.


Whenever a request arrives, we have to decide if and how to serve it. Depending on how we serve   request $i$, say  $x_i$, we gain a certain value $v_i(x_i)$. Our task is to maximize the sum of values over all requests $\sum_{i=1}^n v_i(x_i)$.  Our algorithm $\text{Alg}$ includes a time-dependent payment component. The payment that request $i$ has to make is the product of a time-dependent \emph{discount} function $\alpha(t)$ and a \emph{base price} $\baseprice(x_i)$. \Thomas{The base price depends on the allocation up to this point and how much the new choice limits other allocations in the future. However, it does not depend on $t$, the time that has passed up to this point.} If request $i$ has to pay $p_i(x_i, \arrivaltime_i)=\alpha(\arrivaltime_i, \baseprice(x_i)$ for our decision $x_i$, then its \emph{utility} is given by $u_i = v_i(x_i) - p_i(x_i, \arrivaltime_i)$. We write $\Util = \sum_{i=1}^n u_i$ for the sum of utilities and $\Rev = \sum_{i=1}^n p_i(x_i, \arrivaltime_i)$ for the sum of payments. The value achieved by  $\text{Alg}$ equals $\Util + \Rev$.

Next we define a residual function that has the interpretation of ``expected remaining value in the instance at time $t$". In Lemma~\ref{lem:residualSuffices} we show that the existence of a residual function for $\text{Alg}$ suffices to give a $(1-1/e)$-approximation prophet secretary.

\begin{definition}[Residual] \label{def:residual}
Consider a prophet secretary problem with expected offline value $\E[\OPT]$. For any algorithm $\text{Alg}$ based on a differentiable \emph{discount} function $\alpha(t)\colon [0,1]\rightarrow [0,1]$, a differentiable function $r(t) \colon [0, 1] \to \mathbb{R}_{\geq 0}$ is called a  \emph{residual} if it satisfies the following three conditions for every choice of $\alpha$.
\begin{subequations}
\begin{align}
r(0) & = \E[\OPT] \label{eq:template:startvalue} \\
\E[\Rev] & \geq - \int_{t=0}^1 \alpha(t) \cdot r'(t) \cdot  dt \label{eq:template:revenue} \\
\E[\Util] & \geq \int_{t=0}^1 (1 - \alpha(t))\cdot  r(t)\cdot  dt. \label{eq:template:utility} 
\end{align}
\end{subequations}
\end{definition}

We would like to remark here that this definition is similar in spirit to balanced thresholds \cite{KW-STOC12} and balanced prices \cite{DFKL-FOCS17}. However, it is different because we have to take into account the random arrivals.

As an illustration of Definition~\ref{def:residual}, consider the case of a single item. That is, we are presented a sequence of $n$ real numbers and may select only up to one of them (previously studied in~\cite{esfandiari2015prophet}).

\begin{example}[Single Item]
\label{example:single-item}
Suppose buyer $i \in [n]$ arrives with random value $v_i$ at time $\arrivaltime_i$ chosen uniformly at random between $0$ and $1$. Define $\baseprice = \E[\max_i v_i]$ as the  \emph{base price} of the single item. A buyer arriving at time $t$ is offered the item at price $\alpha(t) \cdot b$, and she accepts the offer if and only if $v_i \geq \alpha(t) \cdot b$. We show that   $r(t) = \Pr[\text{item not sold before $t$}] \cdot b$ is a \emph{residual} function.

By  definition, \eqref{eq:template:startvalue} holds trivially. To see that \eqref{eq:template:revenue} holds, observe that the increase in revenue from time $t$ to time $t+\epsilon$ is approximately $\alpha(t) \cdot b$ if the item is allocated during this time, and is $0$ otherwise. That is, the expected increase in revenue is approximately $\alpha(t) ( r(t) - r(t+\epsilon) )$. Taking the limit for $\epsilon \to 0$ then implies \eqref{eq:template:revenue}, i.e., $\E[\Rev] = - \int_{t=0}^1 \alpha(t) r'(t) dt$.

For \eqref{eq:template:utility}, we consider the expected utility of a buyer $i$ conditioning on her arriving at time $t$
\begin{align*}
\E[u_i \mid \arrivaltime_i = t] & = \E[ \mathbf{1}_\text{item not sold before $t$} \cdot (v_i - \alpha(t) \cdot b)^+ \mid \arrivaltime_i = t]\\
& = \Pr[\text{item not sold before $t$} \mid \arrivaltime_i = t] \cdot \E[(v_i - \alpha(t) \cdot b)^+].
\end{align*}
Here we use that the event  the item is sold before $t$ does not depend on $v_i$ because buyer $i$ only arrives at time $t$. The expectation in turn only depends on $v_i$. It is also important to observe that $\Pr[\text{item not sold before $t$} \mid \arrivaltime_i = t] \geq \Pr[\text{item not sold before $t$}]$.
Next, we take the sum over all buyers $i$ and use that $\E[\sum_{i=1}^n (v_i - \alpha(t) \cdot b)^+] \geq \E[\max_i (v_i - \alpha(t) \cdot b)] = \E[\max_i v_i] - \alpha(t) \cdot b= (1 - \alpha(t)) \cdot b$ to get 
\begin{align*}
\sum_{i=1}^n \E[u_i \mid \arrivaltime_i = t] 
\quad \geq \quad \Pr[\text{item not sold before $t$}] \cdot (1 - \alpha(t)) \cdot b \quad = \quad (1 - \alpha(t)) \cdot r(t).
\end{align*}
This implies
\[
\E[\Util] = \sum_{i=1}^n \int_{t=0}^1 \E[u_i \mid \arrivaltime_i = t] dt = \int_{t=0}^1 \sum_{i=1}^n \E[u_i \mid \arrivaltime_i = t] dt \geq \int_{t=0}^1 (1 - \alpha(t))\cdot r(t) \cdot dt.
\]

\end{example}
%
%
%

We now use the properties of a residual function  to design a $(1 - {1}/{e})$-approximation algorithm. To this end, we choose $\alpha(t)$ in a manner that makes the sum of the expected revenue and buyers' utilities independent of $r(t)$. This allows us to compute expected welfare, even though we cannot compute $r(t)$ directly.


\begin{lemma}\label{lem:residualSuffices}
For a prophet secretary problem, if there exists a residual function $r(t)$ for algorithm $\text{Alg}$ as defined in Definition~\ref{def:residual}, then setting $\alpha(t)= 1-e^{t-1}$ gives  a $(1 - {1}/{e})$-approximation.
\end{lemma}

\begin{proof}
To further simplify Eq.~\eqref{eq:template:revenue}, we observe that applying integration by parts  gives
\[
\int r'(t)\cdot \alpha(t) \cdot dt = r(t)\cdot \alpha(t) - \int r(t) \alpha'(t) \cdot dt.
\]
So in combination
\begin{align}
\E[\Rev] \geq - \left( \left[ r(t) \cdot \alpha(t) \right]_{t=0}^{1} -  \int_{t=0}^{1} r(t) \cdot\alpha'(t) \cdot dt \right).
\label{eq:template:revenue2}
\end{align}
Now adding \eqref{eq:template:revenue2} and \eqref{eq:template:utility} gives,
\begin{align*}
\E[\text{Alg}] &= \E[\Util] + \E[\Rev] \\
& \geq \int_{t=0}^1 r(t) \cdot (1 - \alpha(t))  \cdot dt - \left[ r(t) \alpha(t) \right]_{t=0}^{1} +  \int_{t=0}^{1} r(t)\alpha'(t) \cdot dt\\
&= \int_{t=0}^1 r(t) \cdot (1 - \alpha(t)+\alpha'(t))  \cdot dt - \left[ r(t) \alpha(t) \right]_{t=0}^{1} \enspace .
\end{align*}
Although we do not know $r(t)$ and computing $\int_{t=0}^1 r(t) \cdot (1 - \alpha(t)+\alpha'(t))  \cdot dt$ seems difficult, we have the liberty of selecting the function $\alpha(t)$. By choosing  $\alpha(t)$  satisfying $1 - \alpha(t)+\alpha'(t) = 0$ for all $t$, this integral becomes independent of $r(t)$ and simplifies to $0$. In particular, let $\alpha(t)= 1-e^{t-1}$. This gives,
\begin{align*}
\E[Alg] \quad &\geq \quad - \left[ r(t)\cdot \alpha(t) \right]_{t=0}^{1} \quad \\ &= \quad \left( 1-\frac1e \right) r(0) \quad \\ &= \quad \left( 1-\frac1e \right) \E[\OPT] \enspace . 
\end{align*}
\end{proof}

\section{Prophet Secretary for Combinatorial Auctions} \label{section:scaling}

Let $N$ denote a set of $n$ buyers and $M$ denote the set of $m$ indivisible items. Suppose buyer $i$ arrives at a time $\arrivaltime_i$ chosen uniformly at random between $0$ and $1$. Let $\val_i\colon 2^M \to \mathbb{R}_{\geq 0}$ (similarly $\hat{\val}_i)$ denote the random combinatorial valuation function of buyer $i$. \Thomas{In order to ensure polynomial running times, w}e assume that the distribution of $\val_i$ has a polynomial support $ \{\val_i^1,  \val_i^2, \ldots, \}$, where $\sum_k \Pr[\val_i = \val_i^k] = 1$. Note that this assumption only simplifies notation. If we only have sample access to the distributions, then we can replace $\{\val_i^1,  \val_i^2, \ldots, \}$ by an appropriate number of samples. \Thomas{Within our proofs, we will use $\hat{\bval}$ to denote an independent, fresh sample from the distribution.}

By $\barrivaltime$ and $\bval$ (similarly $\hat{\bval}$)  we denote the vector of all the buyer arrival times and valuations, respectively. Also, let $\bval_{-i}$ (similarly $\hat{\bval}_{-i}$) denote valuations of all buyers except buyer $i$. For the special case of single items, we let $\val_{ij}$ denote $\val_i(\{j\})$. Let $q_j(t)$ denote the probability that item $j$ has not been sold before time $t$, where the probability is over valuations $\bval$, arrival times $\barrivaltime$, and any randomness of the algorithm.



\subsection{Bipartite Matching}\label{sec:matching}
In the bipartite matching setting  all buyers are unit-demand, i.e. $\val_i(S) = \max_{j\in S} \val_{ij}$. We can therefore assume that no buyer buys more than one item. We restate our result.

\bipMatching*

To define prices of items, let \emph{base price} $\baseprice_j$ denote the expected value of the buyer that buys item $j$ in the offline welfare maximizing allocation (maximum weight matching).  
Now consider an algorithm that prices item $j$ at  $\alpha(t) \cdot \baseprice_j$  at time $t$ and allows the incoming buyer to pick any of the unsold items; here $\alpha(t)$ is a continuous differentiable discount function.

Consider the function $r(t)=\sum_j q_j(t) \cdot b_j$. Clearly, $r(0) = \E[\OPT]$. Using the following  Lemma~\ref{lem:utilMatching} and Claim~\ref{claim:revMatching}, we prove  that $r$ is a residual function for our algorithm.  Since the algorithm is clearly incentive-compatible,   Lemma~\ref{lem:residualSuffices} implies Theorem~\ref{thm:bipMatching}.

\begin{lemma} \label{lem:utilMatching}
We can lower bound the total expected utility by
\begin{align} 
\E_{\bval,\barrivaltime}[\Util] \geq \sum_{j} \int_{t=0}^1 q_j(t)  \cdot (1 - \alpha(t) ) \cdot \baseprice_j \cdot dt. \label{eq:utilBipMatching}
\end{align}
\end{lemma}

\begin{proof}
Since   buyer $i$ arriving at time $t$ can pick any of the unsold items, we have
\begin{align*}
\E_{\bval,\barrivaltime}[u_i \mid \arrivaltime_i = t] = \E_{\bval} \left[ \max_{j} \mathbf{1}_{\text{$j$ not sold before $t$}} \cdot \left(\val_{i,j} - \alpha(t) \cdot \baseprice_j\right)^+ \growingmid \arrivaltime_i = t \right].
\end{align*}
One particular choice of buyer $i$ is to choose item $OP\arrivaltime_i(\val_i, \hat{\bval}_{-i})$ if it is still available, and no item otherwise. This gives us a lower bound of
\begin{align*}
\E_{\bval,\barrivaltime}[u_i \mid \arrivaltime_i = t] & \geq \E_{\bval} \left[ \mathbf{1}_{\text{$OP\arrivaltime_i(\val_i, \hat{\bval}_{-i})$ not sold before $t$}} \cdot \left(\val_{i, OP\arrivaltime_i(\val_i, \hat{\bval}_{-i})} - \alpha(t) \cdot \baseprice_j\right)^+ \growingmid \arrivaltime_i = t \right] \\
&= \sum_{j} \E_{\bval, \hat{\bval}} \left[ \mathbf{1}_{\text{$j$ not sold before $t$}} \cdot \mathbf{1}_{j = OP\arrivaltime_i(\val_i, \hat{\bval}_{-i})} \cdot \left(\val_{i,j} - \alpha(t) \cdot \baseprice_j\right)^+ \growingmid \arrivaltime_i = t \right].
\end{align*}
Note that in the product, the fact whether $j$ is sold before $t$ only depends on $\bval_{-i}$ and the arrival times of the other buyers. It does not depend on $\val_i$ or $\hat{\bval}$. The remaining terms, in contrast, only depend on $\val_i$ and $\hat{\bval}_{-i}$. Therefore, we can use independence to split up the expectation and get
\begin{align*}
& \E_{\bval,\barrivaltime}[u_i \mid \arrivaltime_i = t] \\
& \geq \sum_{j} \Pr[\text{$j$ not sold before $t$} \mid \arrivaltime_i = t] \cdot \E_{\val_i, \hat{\bval}_{-i}} \left[ \mathbf{1}_{j = OP\arrivaltime_i(\val_i, \hat{\bval}_{-i})} \cdot (\val_{i,j} - \alpha(t) \cdot \baseprice_j)^+ \growingmid \arrivaltime_i = t \right].
\end{align*}
Next, we use that $\Pr[\text{$j$ not sold before $t$} \mid \arrivaltime_i = t] \geq q_j(t)$ by Lemma~\ref{lem:withwithouiBip} and that $\val_i$ and $\hat{\val}_i$ are identically distributed. Therefore, we can swap their roles inside the expectation. Overall, this gives us
\begin{align}
\E_{\bval,\barrivaltime}[u_i \mid \arrivaltime_i = t] 
&\geq \sum_{j} q_j(t) \cdot \E_{\hat{\bval}} \left[ \mathbf{1}_{j = OP\arrivaltime_i(\hat{\bval})} \cdot (\hat{\val}_{i, j} - \alpha(t) \cdot \baseprice_j) \right].
\label{eq:matching:singleutility}
\end{align}
Next, observe that $\E_{\hat{\bval}}[\sum_i \mathbf{1}_{j = OP\arrivaltime_i(\hat{\bval})} \cdot \hat{\val}_{i, j}] = \baseprice_j$ by the definition of $\baseprice_j$. Therefore, using linearity of expectation, summing up \eqref{eq:matching:singleutility} over all buyers $i$ gives us
\[
 \sum_i \E_{\bval,\barrivaltime}[u_i \mid \arrivaltime_i = t] \geq q_j(t)  \cdot (1 - \alpha(t) ) \cdot \baseprice_j.
\]

\noindent Now, taking the expectation over $t$, we get
\begin{align*}
\E_{\bval,\barrivaltime}\left[\sum_i u_i\right] &= \sum_i \int_{t=0}^1 \E_{\bval,\barrivaltime}[u_i \mid \arrivaltime_i = t] \cdot dt \\ &= \int_{t=0}^1 \sum_i \E_{\bval,\barrivaltime}[u_i \mid \arrivaltime_i = t] \cdot dt  \\
& \geq \int_{t=0}^1 \sum_{j} q_j(t)  \cdot (1 - \alpha(t) ) \cdot \baseprice_j \cdot dt \\ &= \sum_{j} \int_{t=0}^1 q_j(t)  \cdot (1 - \alpha(t) ) \cdot \baseprice_j \cdot dt \enspace . 
\end{align*}
\end{proof}

We next give a bound on the revenue generated by our algorithm.
\begin{claim} \label{claim:revMatching}
 We can  bound the total expected revenue by
\begin{align} \label{eq:bipMatchRev}
	\E_{\bval,\barrivaltime}[\Rev] =   - \sum_j \int_{t=0}^{1} q'_j(t) \alpha(t) \cdot \baseprice_j \cdot dt.
\end{align}
\end{claim}
\begin{proof}
Since $-q'_j(t)  dt$ is the  probability that item $j$ is bought between $t$ and $t+dt$ (note $q_j(t)$ is decreasing in $t$), we have
\begin{align*}
\E[\Rev] &= - \sum_j \int_{t=0}^{1} q'_j(t) \alpha(t) \cdot \baseprice_j \cdot dt \enspace . 
\end{align*}
\end{proof}

Finally, we prove the missing lemma that removes the conditioning on the arrival time. 

\begin{lemma}\label{lem:withwithouiBip} We have
\begin{align*} 
\E_{\bval_{-i}} \left[ \Pr_{\barrivaltime}[\text{$j$ not sold before $t$} \mid \arrivaltime_i = t] \right] \geq q_j(t).
\end{align*}
\end{lemma}

\Thomas{The idea is that if buyers arrive earlier in the process, this only \emph{reduces} the available items. It can never happen that such a change makes an item available at a later point. For a single item this is trivial, for multiple items and other combinatorial valuations it does not necessarily hold.}

\ifFULL
\begin{proof}
Consider the execution of our algorithm on two sequences that only differ in the arrival time of buyer $i$. To this end, let $\bval$ be arbitrary values and $\barrivaltime$ be arbitrary arrival times. Let $A_{t'}$ be the set of items that are sold before time $t'$ on the sequence defined by $\bval$ and $\barrivaltime$. Furthermore, let $B_{t'}$ be the set of items sold before time $t'$ if we replace $\arrivaltime_i$ by $t$. Ties are broken in the same way in both sequences.

We claim that $B_{t'} \subseteq A_{t'}$ for all $t' \leq t$.

To this end, we observe that by definition $B_{t'} = A_{t'}$ for $t' \leq \min\{ \arrivaltime_i, t \}$ because the two sequences are identical before $\min\{ \arrivaltime_i, t \}$. This already shows the claim for $\arrivaltime_i \geq t$. Otherwise, assume that there is some $t' \leq t$ for which $B_{t'} \not\subseteq A_{t'}$. Let $t_{\inf}$ be the infimum among these $t'$. It has to hold that some buyer $i'$ arrives at time $t_{\inf}$ and buys item $j_A \not\in A_{t_{\inf}}$ in the original sequence and $j_B \not\in B_{t_{\inf}}$ in the modified sequence. Furthermore, we now have to have $B_{t_{\inf}} \not\subseteq A_{t_{\inf}}$ because $t_{\inf}$ was defined to be the infimum of all $t'$ for which $B_{t'} \subseteq A_{t'}$ is not fulfilled. Therefore, $j_A \not\in B_{t_{\inf}}$. Additionally, $j_B \not\in A_{t_{\inf}}$. The reason is that for any $t' < t_{\inf}$ before the next arrival $B_{t'} = B_{t_{\inf}} \cup \{ j_B \}$.

Overall this means that in both sequences at time $t_{\inf}$ buyer $i'$ has the choice between $j_B$ and $j_A$. As his values are identical and ties are broken the same way, it has to hold that $j_B = j_A$, which then contradicts that $B_{t_{\inf}} \not\subseteq A_{t_{\inf}}$.

Taking the expectation over both $\bval$ and $\barrivaltime$, we get
\[
\Pr_{\barrivaltime, \bval}[j \not\in A_t] \leq \Pr_{\barrivaltime, \bval}[j \not\in B_t].
\]
This implies the Lemma~\ref{lem:withwithouiBip} because
\[
\Pr_{\barrivaltime, \bval}[\text{$j$ not sold before $t$}] = \Pr_{\barrivaltime, \bval}[j \not\in A_t]
\]
\[
\Pr_{\barrivaltime, \bval}[\text{$j$ not sold before $t$} \mid \arrivaltime_i = t] = \Pr_{\barrivaltime, \bval}[j \not\in B_t].
\]
\end{proof}
\else
We present a proof of Lemma~\ref{lem:withwithouiBip} in the full version.
\fi



\subsection{XOS Combinatorial Auctions}\label{sec:XOSCombAuctions}
In this section we prove our main result (restated below) for combinatorial auctions.
\CAPS*

Recollect that the random valuation $\val_i$ of every buyer $i$ has a polynomial support. We can therefore write the following  expectation-version of the configuration LP, which gives us an upper bound on the expected offline social welfare.
\begin{align*}
\max & \sum_i \sum_k \sum_S \val_i^k(S) \cdot x_{i, S}^k \\
\text{s.t. } & \sum_i \sum_k \sum_{S: j \in S} x_{i, S}^k \quad = \quad 1 && \text{ for all $j \in M$}\\
& \sum_S x_{i, S}^k \quad = \quad \Pr[\val_i = \val_i^k] && \text{ for all $i, k$}
\end{align*}

The above configuration LP can be solved with a polynomial number of calls to  demand oracles of buyer valuations (see~\cite{DNS-MOR10}).
Since all functions $\val_i^k$ are XOS, there exist additive supporting valuations; that is, there exist numbers $\val_{i, j}^{k, S} \geq 0$ s.t. $\val_{i, j}^{k, S} = 0$ for $j \not\in S$, $\sum_{j \in S} \val_{i, j}^{k, S} = \val_i^k(S)$\Thomas{, and $\sum_{j \in S'} \val_{i, j}^{k, S} \leq \val_i^k(S')$ for all $S'$}. Before describing our algorithm, we define a \emph{base price} for every item.

\begin{definition}
The \emph{base price} $\baseprice_j$ of every item $j\in M$ is $\sum_{i,k} \sum_{S:j\in S} \val_{i, j}^{k, S} x_{i,S}^k$.
\end{definition}

Since $\sum_S x_{i, S}^k = \Pr[\val_i = \val_i^k]$, consider an algorithm that on arrival of buyer $i$ with valuation $\val_i^k$ draws an independent random set $S$ with probability ${x_{i, S}^k}/{ \Pr[\val_i = \val_i^k]}$. Let $S_i^\ast$ denote this drawn set. 
This distribution also satisfies  that for every item $j$,
\begin{align} \label{eq:XOSBasePriceAltForm}
\sum_i \E_{\val_i,S_i^\ast}\left[ \mathbf{1}_{j\in S_i^\ast} \cdot \val_{i, j}^{k, S_i^\ast} \right] \quad = \quad \sum_{i,k} \Pr[\val_i = \val_i^k] \cdot \sum_{S: j\in S} \frac{ x_{i, S}^k }{\Pr[\val_i = \val_i^k]} \cdot \val_{i,j}^{ k,S_i^\ast } \quad = \quad  \baseprice_j. 
\end{align}


\noindent Now consider the supporting additive valuation for $S_i^\ast$ in the XOS valuation function $\val_i^k$ of buyer  $i$. This can be found using the XOS oracle for $\val_i^k$~\cite{DNS-MOR10}.
Our algorithm assigns her every item $j$ \Thomas{that has not been allocated so far and} for which $\val_{i, j}^{k, S_i^\ast} \geq \alpha(t) \cdot \baseprice_j$, where $\alpha(t)$ is a continuous differentiable function of $t$.  Note that since we do not allow buyer $i$ to choose items outside set $S_i^\ast$, the mechanism defined by this algorithm need not be incentive compatible.

Consider the function $r(t)=\sum_j q_j(t) \cdot b_j$, where again $q_j(t)$ denotes the probability that item $j$ has not been sold before time $t$. Clearly, $r(0) = OPT$. Using the following  Lemma~\ref{lem:utilXOS} and Claim~\ref{claim:revXOS}, we prove that  $r$ is a residual function  for our algorithm. Hence,   Lemma~\ref{lem:residualSuffices} implies Theorem~\ref{thm:XOS}.

\begin{lemma} \label{lem:utilXOS}
The expected utility of the above algorithm is lower bounded by
\begin{align} 
\E_{\bval,\barrivaltime}[\Util] \geq \sum_{j} \int_{t=0}^1 q_j(t)  \cdot (1 - \alpha(t) ) \cdot \baseprice_j \cdot dt. \label{eq:utilXOS}
\end{align}
\end{lemma}

\begin{proof} Given that buyer $i$ arrives at $t$ and only buys item $j$  if $\val_{i, j}^{k, S_i^\ast} \geq \alpha(t) \cdot \baseprice_j$, her utility is
\begin{align*}
\E_{\bval,\barrivaltime}[u_i \mid \arrivaltime_i = t] &= \sum_{j} \E_{\bval,\barrivaltime,S_i^\ast} \left[ \mathbf{1}_{\text{$j$ not sold by $t$}} \cdot \mathbf{1}_{j \in S_i^\ast} \cdot \left(\val_{i, j}^{k,S_i^\ast} - \alpha(t) \cdot \baseprice_j\right)^+ \growingmid \arrivaltime_i = t \right] 
\end{align*}
Using the fact that whether $j$ is sold before $t$ only depends on $\bval_{-i}$ and $\barrivaltime$, and not on $\val_i$ or $S_i^\ast$, 
\begin{align*}
\E_{\bval,\barrivaltime}[u_i \mid \arrivaltime_i = t]  
&= \sum_{j} \Pr_{\bval_{-i},\barrivaltime}[\text{$j$ not sold by $t$} \mid \arrivaltime_i = t] \cdot \E_{\val_i,S_i^\ast}\left[ \mathbf{1}_{j \in S_i^\ast} \cdot \left(\val_{i, j}^{k,S_i^\ast} - \alpha(t) \cdot \baseprice_j\right)^+   \right]. 
\end{align*}
Now, observe that in our algorithm every buyer $i$ independently decides which set of items $S_i^\ast$ it will attempt to buy.  Crucially, the probability of an item $j$ being sold by time $t$ can only increase if more buyers arrive before $t$.  Therefore, 
\[ \Pr_{\bval_{-i},\barrivaltime}[\text{$j$ not sold by $t$} \mid \arrivaltime_i = t]  \quad \geq \quad \Pr_{\bval,\barrivaltime}[\text{$j$ not sold by $t$}] \quad =  \quad q_j(t). 
\]
Thus, we  get
\begin{align*} 
\E_{\bval,\barrivaltime}[u_i \mid \arrivaltime_i = t] &\geq \sum_j q_j(t) \cdot \E_{\val_i,S_i^\ast} \left[ \mathbf{1}_{j \in S_i^\ast} \cdot \left(\val_{i, j}^{k, S_i^\ast} - \alpha(t) \cdot \baseprice_j\right)^+  \right] \\
&\geq \sum_j q_j(t) \cdot \E_{\val_i,S_i^\ast} \left[ \mathbf{1}_{j \in S_i^\ast} \cdot \left(\val_{i, j}^{k, S_i^\ast} - \alpha(t) \cdot \baseprice_j\right)  \right].
\end{align*}
Finally, recollect from Eq.~\eqref{eq:XOSBasePriceAltForm} that $\sum_i \E_{\val_i,S_i^\ast} \left[ \mathbf{1}_{j \in S_i^\ast} \cdot \val_{i, j}^{k, S_i^\ast} \right]= \baseprice_j$. Moreover, 
\[ \sum_i \E_{\val_i,S_i^\ast}\left[ \mathbf{1}_{j\in S_i^\ast} \right] \quad = \quad \sum_{i,k} \Pr[\val_i = \val_i^k] \cdot \sum_{S: j\in S} \frac{ x_{i, S}^k }{\Pr[\val_i = \val_i^k]}  \quad = \quad  1. 
\]
Hence, by linearity of expectation
\begin{align*}
\sum_i \E[u_i \mid \arrivaltime_i = t] \geq \sum_{j} q_j(t) \cdot (1 - \alpha(t)) \cdot b_j.
\end{align*}
\end{proof}

We next give a bound on the revenue generated by our algorithm.
\begin{claim} \label{claim:revXOS}
 We can  bound the total expected revenue by
\begin{align} \label{eq:XOSRev}
	\E_{\bval,\barrivaltime}[\Rev] =   - \sum_j \int_{t=0}^{1} q'_j(t) \alpha(t) \cdot \baseprice_j \cdot dt.
\end{align}
\end{claim}
\begin{proof}
Since $-q'_j(t)  dt$ is the  probability that item $j$ is bought between $t$ and $t+dt$ (note $q_j(t)$ is decreasing in $t$), we have
\begin{align*}
\E[\Rev] &= - \sum_j \int_{t=0}^{1} q'_j(t) \alpha(t) \cdot \baseprice_j \cdot dt \enspace . 
\end{align*}
\end{proof}

\IGNORE{
\begin{align*}
\E[u_i \mid \arrivaltime_i = t] & \geq \sum_{j} \E \left[ \mathbf{1}_{\text{$j$ not sold by $t$}} \mathbf{1}_{j \in S^\ast(\val_i, \bval^{(i)}_{-i})} \cdot \left(\val_{i, j}^{k, S^\ast(\val_i, \bval^{(i)}_{-i})} - \alpha(t) \cdot \baseprice_j\right)^+ \growingmid \arrivaltime_i = t \right] \\
& = \sum_{j} \Pr[\text{$j$ not sold by $t$} \mid \arrivaltime_i = t] \E\left[ \mathbf{1}_{j \in S^\ast(\val_i, \bval^{(i)}_{-i})} \cdot \left(\val_{i, j}^{k, S^\ast(\val_i, \bval^{(i)}_{-i})} - \alpha(t) \cdot \baseprice_j\right)^+  \growingmid \arrivaltime_i = t \right] \\
& \geq \sum_{j} q_j(t) \E\left[ \mathbf{1}_{j \in S^\ast(\val_i, \bval^{(i)}_{-i})} \cdot \left(\val_{i, j}^{k, S^\ast(\val_i, \bval^{(i)}_{-i})} - \alpha(t) \cdot \baseprice_j\right) \right] \\
& = \sum_{j} q_j(t) \E\left[ \mathbf{1}_{j \in S^\ast(v)} \cdot \left(\val_{i, j}^{k, S^\ast(v)} - \alpha(t) \cdot \baseprice_j\right) \right]
\end{align*}]}

\section{Prophet Secretary for Matroids}\label{sec:matroid}

Let $\val_i $ denote the random value of the $i$'th buyer (element) and let $\hat{\val_i}$ denote another independent draw from the value distribution of the $i$'th buyer.
 The problem is to select a subset $I$ of the buyers that form a feasible set in matroid $\M$, while trying to maximize $\sum_{i\in I} \val_i$.  We restate our main result for the matroid setting.
\MPS*

We need the following notation to describe our algorithm.
\begin{definition} \label{def:RemCost}
For a given \Sahil{vector $\hat{\bval}$ of values of $n$ items} and $A\subseteq [n]$, we define the following:
\begin{itemize}
\item Let  $Opt(\hat{\bval} \mid A) \Thomas{\subseteq [n] \setminus A}$ denote the optimal solution set  in the contracted matroid $\M/A$.
\item Let  $R(A,\hat{\bval}) := \sum_{i\in Opt(\hat{\bval} \mid A)} \hat{\bval}_i $  denote the remaining value after selecting set $A$.
\end{itemize}
\end{definition}

We next define a \emph{base price} of for every buyer $i$.
\begin{definition} Let $A$ denote the independent set of buyers that have been accepted till now.
\begin{itemize}
\item Let $\baseprice_i(A,\hat{\bval}) := R(A,\hat{\bval}) - R(A\cup \{ i \},\hat{\bval})$  denote a threshold for buyer $i$.
\item Let $\baseprice_i(A) := \E_{\hat{\bval}}[\baseprice_i(A,\hat{\bval})]$  denote the \emph{base price} for buyer $i$.
\end{itemize}

\end{definition}

\Thomas{Starting with $A_0=\emptyset$, let $A_t$ denote the set of accepted buyers \emph{before} time $t$. This is a random variable that depends on the values $\bval$ and arrival times $\barrivaltime$.}
Suppose a  buyer $i$ arrives at time $t$, then our algorithm selects $i$ iff both $\val_i > \alpha(t) \cdot \baseprice_i(A_t) $ and selecting $i$ is feasible in $\M$.

Consider the function $r(t) := \E_{\bval, \hat{\bval},\barrivaltime}[R(A_t,\hat{\bval})]$, where $A_t$ is a function of $\bval$ \Thomas{and $\barrivaltime$}. Clearly, $r(0) = \E[\OPT]$. Using the following  Lemma~\ref{lem:utilMat} and Claim~\ref{claim:revMat}, we prove  that $r$ is a residual function. Hence,   Lemma~\ref{lem:residualSuffices} implies Theorem~\ref{thm:matProphetSec}.

\begin{claim}\label{claim:revMat}
\[ \E_{\bval,\barrivaltime}[\Rev] = - \int_{t=0}^1 \alpha(t) \cdot r'(t) dt.
\]
\end{claim}
\IGNORE{
\tnote{This is the original proof.}
\begin{proof} We prove the lemma after conditioning on  $\bval$ and $\barrivaltime$. For any time $t \in [0,1]$, we claim that $d \Rev = \alpha(t) \cdot dE_{\hat{\bval}}[C(A_t,\hat{\bval})]$. This claim proves the lemma because conditioning on $\hat{\bval}$ fixes $Opt(\hat{\bval})$, which implies $d C(A_t,\hat{\bval}) = - d R(A_t,\hat{\bval})$ by Definition~\ref{def:RemCost}. The claim is true  because at any time $t$, either no item is picked  and both sides of the claim are zero. Otherwise, we pick some buyer $i$ (i.e., add $i$ to set $A_t$) that changes the $\Rev$  by $ \baseprice_i = \alpha(t) \cdot \E_{\hat{\bval}}[C(A_t \cup \{ i \},\hat{\bval})-C(A_t ,\hat{\bval})]$. However,  $\E_{\hat{\bval}}[C(A_t \cup \{ i \},\hat{\bval})-C(A_t ,\hat{\bval})]$ is exactly the change in $C(A_t,\hat{\bval})$ when we pick of $i$ at time $t$. Hence, the claim and the lemma is true.
\end{proof}
\tnote{This is an alternative.}
}

\ifFULL
\begin{proof}
Consider the time from $t$ to $t + \epsilon$ for some $t \in [0,1]$, $\epsilon > 0$. Let us fix the arrival times $\barrivaltime$ and values $\bval$ of all elements. This also fixes the sets $(A_t)_{t \in [0,1]}$. Let $i_1, \ldots, i_k$ be the arrivals between $t$ and $t+\epsilon$ that get accepted in this order. Note that it is also possible that $k=0$. The revenue obtained between $t$ and $t+\epsilon$ is now given as
\begin{align*}
\Rev_{\leq t+\epsilon} - \Rev_{\leq t} & = \sum_{j=1}^k \alpha(t_{i_j}) \baseprice_{i_j}(A_{t_{i_j}})\\
& = \sum_{j=1}^k \alpha(t_{i_j}) \E_{\hat{\bval}} \left[ R(A_t \cup \{i_1, \ldots i_{j-1}\}, \hat{\bval}) - R(A_t \cup \{i_1, \ldots i_j\}, \hat{\bval}) \right] \\
& \geq \alpha(t+\epsilon) \E_{\hat{\bval}} \left[ R(A_t, \hat{\bval}) - R(A_{t+\epsilon}, \hat{\bval}) \right].
\end{align*}
Taking the expectation over $\bval$ and $\barrivaltime$, we get by linearity of expectation
\[
\E_{\bval,\barrivaltime}[\Rev_{\leq t+\epsilon}] - \E_{\bval,\barrivaltime}[\Rev_{\leq t}] \geq \alpha(t+\epsilon) (r(t) - r(t+\epsilon)).
\]
By the same argument, we also have
\[
\E_{\bval,\barrivaltime}[\Rev_{\leq t+\epsilon}] - \E_{\bval,\barrivaltime}[\Rev_{\leq t}] \leq \alpha(t) (r(t) - r(t+\epsilon)) .
\]
In combination, we get that
\[
\frac{d}{dt} \E_{\bval,\barrivaltime}[\Rev_{\leq t}] = - \alpha(t) r'(t) ,
\]
which implies the claim.
\end{proof}
\else
We prove Claim~\ref{claim:revMat} in the full version.
\fi

\begin{lemma}\label{lem:utilMat}
\[ \E_{\bval,\barrivaltime}[\text{Utility}] \geq  \int_{t=0}^1  (1-\alpha(t)) \cdot r(t) dt.
\]
\end{lemma}
\begin{proof} The utility of buyer $i$ arriving at time $t$ is given by 
\begin{align*} \E_{\bval,\barrivaltime}[u_i \mid \arrivaltime_i = t] &= \E_{\bval,\barrivaltime_{-i}} \left[ \left(\val_i - \alpha(t) \cdot \baseprice_i(A_t) \right) ^+ \cdot \one_{i \not \in Span(A_t)} \growingmid \arrivaltime_i = t \right] .
\end{align*}
\Thomas{Observe that $A_t$ does not depend on $\val_i$ if $\arrivaltime_i = t$ because it includes only the acceptances \emph{before} $t$. It does not depend on $\hat{\val_i}$ either, as $\hat{\val_i}$ is only used for analysis purposes and not known to the algorithm. Since $\val_i$ and $\hat{\val_i}$ are identically distributed, we can also write}
\begin{align}\label{eq:AtXtIndependence}
\E_{\bval,\barrivaltime}[u_i \mid \arrivaltime_i = t] = \E_{\bval,\hat{\bval},\barrivaltime_{-i}} \left[ \left(\hat{\val}_i - \alpha(t) \cdot \baseprice_i(A_t) \right) ^+ \cdot \one_{i \not \in Span(A_t)} \growingmid \arrivaltime_i = t \right].
\end{align}
Now observe that buyer $i$ can belong to $Opt(\hat{\bval}\mid A_t)$ only if it's not already in $Span(A_t)$, which implies  $\one_{i \not \in Span(A_t)} \geq \one_{i \in Opt(\hat{\bval}\mid A_t)}$. Using this and removing non-negativity, we get
\[ \E_{\bval,\barrivaltime}[u_i \mid \arrivaltime_i = t] \geq \E_{\bval,\hat{\bval},\barrivaltime_{-i}} \left[ \left(\hat{\val}_i - \alpha(t) \cdot \baseprice_i(A_t) \right)  \cdot \one_{i \in Opt(\hat{\bval}\mid A_t)} \growingmid \arrivaltime_i = t \right] .
\]
Now we use Lemma~\ref{lem:withwithouiMatroid}  to remove the conditioning on buyer $i$ arriving at time $t$ as this gives a valid lower bound on expected utility, 
\begin{align} \label{eq:utilBuyerIMatroid}
 \E_{\bval,\barrivaltime}[u_i \mid \arrivaltime_i = t] \geq \E_{\bval,\hat{\bval},\barrivaltime} \left[\left(\hat{\val}_i - \alpha(t) \cdot \baseprice_i(A_t) \right)  \cdot \one_{i\in Opt(\hat{\bval}\mid A_t)} \right].
\end{align}
We can now lower bound sum of buyers' utilities using  Eq.~\eqref{eq:utilBuyerIMatroid} to get
\begin{align*}
\E_{\bval,\barrivaltime}[\text{Utility}] & = \sum_i \int_{t=0}^1 \E_{\bval,\barrivaltime}[u_i \mid \arrivaltime_i = t] \cdot  dt \\
&\geq \sum_i \int_{t=0}^1 \E_{\bval,\hat{\bval},\barrivaltime} \left[\left(\hat{\val}_i - \alpha(t) \cdot \baseprice_i(A_t) \right)  \cdot \one_{i\in Opt(\hat{\bval}\mid A_t)} \right] \cdot dt.
\end{align*}

By moving the sum over buyers inside the integrals, we get
\begin{align*}
\E_{\bval,\barrivaltime}[\text{Utility}]  &\geq  \int_{t=0}^1  \E_{\bval,\hat{\bval},\barrivaltime} \left[ \sum_i \left(\hat{\val}_i - \alpha(t) \cdot \baseprice_i(A_t) \right)  \cdot \one_{i\in Opt(\hat{\bval}\mid A_t)} \right] \cdot dt\\
& =\int_{t=0}^1  \E_{\bval,\hat{\bval},\barrivaltime} \left[   R(A_t,\hat{\bval}) -  \alpha(t) \cdot \sum_{i\in Opt(\hat{\bval}\mid A_t)}  \baseprice_i(A_t)   \right] \cdot dt .
\end{align*}

\IGNORE{\Sahil{ 
Now using Lemma~\ref{lem:KWProp2} for $V=Opt(\hat{\bval}\mid A_t)$, gives
\[	\sum_{i\in Opt(\hat{\bval}\mid A_t)}  \baseprice_i(A_t) \quad \leq \quad \E_{\tilde{\bval}} \left[ R(A_t,\tilde{\bval})\right] \quad \leq  \quad
\E_{\hat{\bval}} \left[ R(A_t,\hat{\bval}) \right] ,
\]
where the last inequality uses the definition of $A_t$.
\[ \E_{\bval,\barrivaltime}[\text{Utility}]  \geq \int_{t=0}^1  \E_{\bval,\hat{\bval},\barrivaltime} \left[ \left(1 - \alpha(t) \right) \cdot R(A_t,\hat{\bval})  \right] \cdot dt.
\]
}}

Finally, using Lemma~\ref{lem:KWProp2} for $V=Opt(\hat{\bval}\mid A_t)$, we get 
\[ \E_{\bval,\barrivaltime}[\text{Utility}]  \geq \int_{t=0}^1  \E_{\bval,\hat{\bval},\barrivaltime} \left[ \left(1 - \alpha(t) \right) \cdot R(A_t,\hat{\bval})  \right] \cdot dt.
\]

\end{proof}


Finally, we prove the missing lemma that removes the conditioning on item $i$ arriving at $t$.
\begin{lemma}\label{lem:withwithouiMatroid} For any $i$, any time $t$, and any fixed $\bval,\hat{\bval}$,  we have
\begin{align*} 
\E_{\barrivaltime_{-i}} \left[ \left(\hat{\val}_i - \alpha(t) \cdot \baseprice_i(A_t) \right) \cdot \one_{i\in Opt(\hat{\bval}\mid A_t)} \mid \arrivaltime_i = t \right]  ~\geq~ \E_{\barrivaltime} \left[\left(\hat{\val}_i - \alpha(t) \cdot \baseprice_i(A_t) \right)  \cdot \one_{i\in Opt(\hat{\bval}\mid A_t)} \right].
\end{align*}
\end{lemma}
\begin{proof}
We prove the lemma for any fixed $\barrivaltime_{-i}$. Suppose we draw a uniformly random $\arrivaltime_i\in [0,1]$. Observe that  if $\arrivaltime_i \geq t$ then we have equality in the above equation because set $A_t$  is the same both with and without $i$. This is also the case when $\arrivaltime_i<t$ but $i$ is not selected into $A_t$. Finally, when $\arrivaltime_i<t$ and $i \in A_t$  we have $\one_{i\in Opt(\hat{\bval}\mid A_t)} =0$ in the presence of item $i$ (i.e., RHS of lemma), making the inequality trivially true.
\end{proof}

\begin{lemma}\label{lem:KWProp2} For any fixed $\bval,\barrivaltime$, time $t$, and set of elements $V$ that is independent in the matroid $\M/A_t$, we have
\[
\sum_{i\in V}  \baseprice_i(A_t) \quad \leq \quad \E_{\hat{\bval}} \left[ R(A_t,\hat{\bval})\right] .
\]
\end{lemma}

\begin{proof}
By definition
\[
\sum_{i\in V}  \baseprice_i(A_t) = \E_{\hat{\bval}} \left[\sum_{i\in V} \left( R(A_t ,\hat{\bval})-R(A_t \cup \{ i \},\hat{\bval})  \right) \right].
\]
Fix the values $\hat{\bval}$ arbitrarily, we also have
\[
\sum_{i\in V} \left( R(A_t ,\hat{\bval})-R(A_t \cup \{ i \},\hat{\bval}) \right) \leq R(A_t,\hat{\bval}).
\]
This follows from the fact that $R(A_t ,\hat{\bval})-R(A_t \cup \{ i \},\hat{\bval})$ are the respective critical values of the greedy algorithm on $\M/A_t$ with values $\hat{\bval}$. Therefore, the bound follows from Lemma~3.2 in \cite{LucierBorodinSODA10}. An alternative proof is given as Proposition~2 in \cite{KW-STOC12} while in our case the first inequality can be skipped and the remaining steps can be followed replacing $A$ by $A_t$.

Taking the expectation over $\hat{\bval}$, the claim follows.
%
%
%
%
\end{proof}

\section{Fixed Threshold Algorithms} \label{section:fixed}
In this section we discuss  the powers and limitations of \textit{Fixed-Threshold Algorithms} (\textit{FTAs}) for single item prophet secretary. In an FTA, we set a fixed threshold for the item at the beginning of the process and then assign it to the first buyer whose valuation exceeds the threshold. The motivation to study FTAs comes from their simplicity, transparency, and fairness in the design of a posted price mechanism (see, e.g., \cite{feldman2015combinatorial}).


In Section~\ref{sec:singleProphSec}, we give a $(1-1/e)$-approximation FTA for single-item prophet secretary. {This seemingly contradicts earlier impossibility results (e.g., \cite{feldman2015combinatorial,esfandiari2015prophet}). However, as we show, these impossibility results do not hold in case of continuous distributions or equivalently randomized tie-breaking.} Next, in Section~\ref{sec:HardnesssingleProphSec}, we present an upper bound for FTAs. In particular, we show that there is no FTA, even for identical distributions, with an approximation factor better than $1-1/e$. This indicates the tightness of our algorithm for prophet secretary. Furthermore, in Appendix~\ref{appendix:matching} we generalize these single item ideas to present an alternate $(1-1/e)$-approximation algorithm for bipartite matching prophet secretary.
	
\subsection{Single Item Prophet Secretary}\label{sec:singleProphSec}

Recall, by $\barrivaltime$ and $\bval$  we denote the random vector of all the buyer arrival times and valuations, respectively. Also, $q(t)$ denotes the probability that the item is unsold till time $t$, where the probability is over valuations $\bval$, arrival times $\barrivaltime$, and any randomness of the algorithm.
We show that a fixed threshold algorithm that selects $\tau$ s.t. 
\[ \Pr_{\bval}[\max\{ v_i \} \leq \tau] = \prod_i \Pr[v_i \leq \tau]  = \frac1e\]
gives a $(1-1/e)$-approximation.

\begin{theorem} \label{theorem:singleTalg}
	There exists a $(1-{1}/{e})$-approximation FTA to single-item prophet secretary.
\end{theorem}
	
	\begin{proof}
		Without loss of generality,  assume  all distributions have a finite expectation and a continuous CDF\footnote{This assumption is without loss because the actual CDF can be approximated with arbitrary precision by a continuous function. This approximation corresponds to a randomized tie-breaking in case of point masses.}. As two extreme selections for the threshold, if we set $\tau$ to zero then the FTA selects the first item, and if we set it to infinity then no item will be selected. Therefore, the assumption for the continuity of the distribution function allows us to select a threshold $\tau$ such that the FTA reaches the end of the sequence with an exact probability of $1/e$. This means all of drawn values are below $\tau$ with probability $1/e$. In the remainder, we show that the FTA based on this choice of $\tau$ lead to a $(1-1/e)$-approximation algorithm.
		
		Let $\OPT$ denote  $\max_i\{v_i\}$ and $\Alg$ be a random variable that indicates the value selected by the algorithm, or is zero if no item is selected. The goal is to show $$\E[\Alg]\geq \left(1-\frac{1}{e}\right)\cdot \E[\OPT] .$$
		
		We have $\E[\Alg]=\E[\Rev]+\E[\Util]$. By definition of $\tau$, the algorithm sells the item with probability exactly $1-1/e$; therefore, $\E[\Rev]=\left(1-\frac{1}{e}\right)\tau$. Below, we show  
\begin{align} \label{eq:prophSecSingleUtil}
\E[\Util] \geq \left (1-\frac1e \right) \cdot \E[ (\OPT- \tau)^+] . 
\end{align}
This suffices to prove Theorem~\ref{theorem:singleTalg} because 
\[ \E[\Alg]=\E[\Rev]+\E[\Util] \geq \left(1-\frac{1}{e}\right) \tau + \left(1-\frac1e \right)  \E[ (\OPT- \tau)^+] \geq \left(1-\frac{1}{e}\right) \E[\OPT].
\]
		
We now prove Eq.~\eqref{eq:prophSecSingleUtil}. For the utility, we know
			\begin{align*}
			\E[\Util] &=  \int_{t=0}^1 \sum_{i=1}^{n}  \E[u_i \mid \arrivaltime_i = t] \cdot dt\\
			& = \int_{t=0}^1 \sum_{i=1}^{n} \Pr[\text{item not sold before t}\mid \arrivaltime_i = t] \cdot \E[(v_i -\tau)^+] \cdot dt \\
			& \geq  \int_{t=0}^1 \sum_{i=1}^{n} q(t) \cdot  \E[(v_i -\tau)^+] \cdot dt \\
			&= \sum_{i=1}^{n} \E[(v_i -\tau)^+]  \cdot \int_{t=0}^1  q(t)   \cdot dt,
			\end{align*}
where the inequality uses the observation  $Pr[\text{item not sold before t}\mid \arrivaltime_i = t] \geq Pr[\text{item not sold before $t$}]$.
In the following Lemma~\ref{lem:qtLowerBound}, we show $q(t) \geq \exp(-t)$. This implies  $\int_{t=0}^1  q(t)   \cdot dt \geq 1-\frac1e$, which proves the missing Eq.~\eqref{eq:prophSecSingleUtil} because 
\[	\sum_{i=1}^{n} \E[(v_i -\tau)^+]  \geq \E[ ( \max_i \{v_i\} - \tau)^+]  =  \E[ (\OPT - \tau)^+]. 
\]
\end{proof}

\begin{lemma}	\label{lem:qtLowerBound} For $t\in [0,1]$, we have
\[ q(t) \geq \exp(-t).
\]
\end{lemma}
\begin{proof} 
Observe that 
\[q(t) = \prod_{i=1}^n \Pr[\text{$i$ does not buy the item till $t$}]. \]
Since $i$ gets the item only by arriving before $t$ and having a value above $\tau$, we get 
\[q(t) = \prod_{i=1}^n (1 - t\cdot \Pr[v_i > \tau]) = \exp \left( \sum_{i=1}^n  \ln (1 - t\cdot \Pr[v_i > \tau]) \right).\]
Notice that for $t,x\in [0,1)$, we have $\ln(1-tx) \geq t\cdot \ln(1-x)$. This gives, 
\[	q(t)  \geq \exp \left( t\cdot \sum_{i=1}^n  \ln (1 -  \Pr[v_i > \tau]) \right) = \exp(t \cdot \ln(1/e)) = \exp(-t),
\]
where we use $\prod_{i=1}^n (1 -  \Pr[v_i > \tau]) = \frac1e$ by definition of $\tau$.
\end{proof}


\subsection{Impossibility for IID Prophet Inequalities}\label{sec:HardnesssingleProphSec}

In the following we prove an impossibility result for FTAs for single item prophet secretary. We show this impossibility even for the special case of iid items. For every $n$, we give a common distribution $D$ for every item such that no FTA can achieve an approximation factor better than $1-1/e$. This also implies the tightness of the algorithm discussed in Section~\ref{sec:singleProphSec}.

\begin{theorem}	\label{thm:prophSecFixedLower}
	Any FTA for iid prophet inequality is at most $(1-\frac{1}{e}+O(\frac{1}{n}))$-approximation\footnote{Jose Correa later pointed to us that the lower bound in Theorem~\ref{thm:prophSecFixedLower} also follows from a result in~\cite{CFHOV-EC17}.}.
\end{theorem}

\begin{proof}
	We prove the theorem by giving a hard input instance for every $n$ as follows: every $v_i$ is $n/(e-1)$ with probability $1/n^2$ and is $(e-2)/(e-1)$ otherwise. The expected maximum value of these $n$ items is
	\begin{align*}
		\E[\OPT] &= \left(1-\frac{1}{n^2}\right)^n \cdot\frac{e-2}{e-1} + \left(1-\left(1-\frac{1}{n^2}\right)^n\right)\frac{n}{e-1} \\ &= 1-O\left(\frac{1}{n}\right)\enspace .
	\end{align*}
	
	In this instance, if $\tau < (e-2)/(e-1)$ then the algorithm selects the first item, and if $(e-2)/(e-1)<\tau\leq n/(e-1)$ then the algorithm can only select $n/(e-1)$. In these cases the approximation factor can be at most $(e-2)/(e-1)\approx 0.58$.
	
	Now, note that the CDF of this input distribution is not continuous. Reshaping a discrete distribution function into a continuous one, however, does not change the approximation factor because in the above example we only need a  slight change at the point $(e-2)/(e-1)$ of the CDF. This change gives us a randomness when $\tau=(e-2)/(e-1)$, which is equivalent to flipping a random coin and skipping every item with some probability $p\leq 1-1/n^2$ if the drawn value is $(e-2)/(e-1)$. With this assumption we have
	\begin{align}
		\E[\Alg] &= \sum_{i=1}^{n} p^i \enspace  \E[v_i \cdot \one_{v_i\geq \tau}] \nonumber \\ &= \frac{1-p^n}{1-p}\enspace \E[v_i \cdot \one_{v_i\geq \tau}] \nonumber \\ &= \frac{1-p^n}{1-p}\enspace \left( \left( 1-\frac{1}{n^2}-p \right) \frac{e-2}{e-1} + \frac{1}{n^2}\frac{n}{e-1}\right)\nonumber \\
		&< \frac{1-p^n}{e-1}\left(e-2+\frac{1}{n(1-p)}\right) \enspace . \label{ineq:hardiid}
	\end{align}
	To complete the proof, it suffices to show that the right hand side of Inequality \eqref{ineq:hardiid} is at most $1-1/e+O(1/n)$. To this end, we try to maximize this term based on parameter $c$ where $p=1-c/n$. We can rewrite the right hand side of the inequality as $$\frac{1-(1-\frac{c}{n})^n}{e-1}\left(e-2+\frac{1}{c}\right)\enspace .$$
	If $c = \Theta(n)$ then this term is at most $(e-2+\Theta(1/n))/(e-1) \approx 0.41+O(1/n)$, which is below $1-1/e$ for sufficiently large $n$. Otherwise $c/n \ll 1$ and we can approximate $(1-c/n)^n$ as $e^{-c}+O(1/n)$. This upper bounds Inequality \eqref{ineq:hardiid} by $(1-e^{-c})(e-2+1/c)/(e-1) + O(1/n)$, where the first term is independent of $n$ and is at most $1-1/e$ for different constants $c$; thereby completing the proof.
\end{proof}

We would like to note that the continuity of the CDF of the input distributions is a useful and natural property that can be used by an FTA. This is because making this assumption allows us to design a $(1-1/e)$-approximation algorithm, as shown by Theorem \ref{theorem:singleTalg}, but not assuming this puts a barrier of $1/2$ for any FTA, which is shown in \cite{feldman2015combinatorial,esfandiari2015prophet}. For example, in the above instance the approximation factor without assuming continuity would be at most $(e-2)/(e-1)\approx 0.58$, which is below the $1-1/e\approx 0.63$ claim of Theorem \ref{theorem:singleTalg}. This contradiction is because without this assumption on the input distribution  the algorithm could not set $\tau$ in a way that the probability of selecting an item becomes exactly $1-1/e$.


\medskip
\noindent
{\bf Acknowledgments}.
We thank a number of colleagues for useful discussions. In particular, we are grateful to Matt Weinberg, Bobby Kleinberg, Anupam Gupta, and Hossein Esfandiari. We are also thankful to Jose Correa and Raimundo Juli\'an Saona for comments on our FTAs.


\bibliographystyle{alpha}
\bibliography{bib}

\appendix

\section{Extension of FTAs to Bipartite Matchings} \label{appendix:matching}

Here we study the set of algorithms that use $m$ fixed thresholds $\tau_1,\ldots,\tau_m$ for the items and have a recommendation strategy which at the arrival of every buyer offers her an item at its fixed price. In particular, when buyer $i$ arrives the algorithm recommends an unsold item $k$ to her and she accepts to buy it if $v_{i,k}\geq \tau_k$, i.e. her valuation for the item is greater than or equal to its price.

\begin{theorem}
	For every instance of matching prophet secretary there exists a sequence of fixed thresholds $\tau_1,\ldots,\tau_k$ and a randomized algorithm which is $(1-1/e)$-approximation in expectation.
\end{theorem}

\begin{proof}
	Our general approach is to extend the methods we have for single item FTA's to an algorithm for matchings. This generalization is similar to a reduction from matchings to single items. However, there are some details that we have to consider. We first show a stronger claim than the statement of Theorem \ref{theorem:singleTalg} which holds for a specific inputs class of prophet secretary. Then we propose a randomized algorithm which exploits the single item algorithm for that class in order to find such fixed thresholds that lead to a $(1-1/e)$-approximation algorithm for matchings.

	We note that the analysis of Theorem \ref{theorem:singleTalg} indicates we can find a single threshold such that every item will be seen with probability at least $1-1/e$. More precisely, the analysis shows that the algorithm gets the same approximation factor from the utility of every buyer, i.e., $\E[\Util]\geq (1-1/e)\cdot \sum_{i=1}^{n}\E[v_i \cdot \one_{v_i\geq \tau}]$. Now, if an input instance guarantees $\sum_{i=1}^{n}Pr[v_i>0]\leq 1$ then
	\begin{align*}
		\E[\Rev] &=\left(1-\frac{1}{e}\right)\tau \\ &\geq\left(1-\frac{1}{e}\right)\tau\sum_{i=1}^{n}Pr[v_i>0] \\ &\geq\left(1-\frac{1}{e}\right)\sum_{i=1}^{n}\E[v_i \cdot \one_{v_i<\tau}]\enspace .
	\end{align*}
	This results in \begin{align*}
	\E[\Alg] &=\E[\Rev]+\E[\Util] \\
	&\geq\left(1-\frac{1}{e}\right)\sum_{i=1}^{n}(\E[v_i \cdot \one_{v_i<\tau}]+\E[v_i \cdot \one_{v_i\geq \tau}]) \\
	&= \left(1-\frac{1}{e}\right)\sum_{i=1}^{n}\E[v_i]\enspace .
	\end{align*}
		The following claim formally states the above result.
	\begin{claim} \label{claim:stasum}
		If the input of prophet secretary guarantees $\sum_{i=1}^{n} Pr[v_i>0]\leq 1$, then there exists an FTA such that $$\E[\Alg]\geq \left(1-\frac{1}{e}\right)\sum_{i=1}^{n}\E[v_i] \enspace .$$
	\end{claim}

	Now we demonstrate how the matching problem reduces to the instances describe above. Let us assume we already know thresholds $\tau_1,\ldots,\tau_m$ for the $m$ items. Upon the arrival of buyer $i$ and realizing $v_i$ we use the following algorithm to recommend an item $k$ to buyer $i$. We first calculate probabilities $p_1,\ldots,p_m$ where $p_k(v_i):=Pr_{\bval_{-i}}[(i,k)\in \M(\bval_i\cup \bval_{-i})]$ and $\M(B)$ is the maximum matching \footnote{WLOG we can assume it is unique for every graph.} of a bipartite graph $B$. These are in fact the probabilities of each of those edges belonging to the maximum matching. Then, by drawing a random number $r\in[0,1]$ we select a candidate item $k$ if $\sum_{l=1}^{k-1}p_l<r\leq \sum_{l=1}^{k} p_l$. In this way, we dependently select a candidate such that every $k$ becomes selected with probability $p_k$. Note that the algorithm might sometimes select none of the items, in which cases there will be no candidate. Finally we recommend item $k$ to buyer $i$ if the item is still unsold, and she buys it if $v_{i,k}\geq \tau_k$.

	The above method for candidate selection has a close relationship with the optimum solution. To put it into perspective, let us define a new distribution $\hat{D}_i:\mathbb{R}^m\rightarrow[0,1]$ for every buyer $i$. This distribution is supposed to show the valuations of $i$ on the items when they are selected as candidates. In other words, for every vector $x=\langle x_1,\ldots,x_m\rangle$ in which at most one of $x_k$'s is non-zero we have $Pr_{\hat{v}_i\sim\hat{D}_i}[\hat{v}_i=x] := \E_{v_i\sim D_i}[\one_{v_{i,k}=x_k} \cdot \one_{k\text{ is a candidate}}]$. Equivalently, $\hat{D}_i$ can be interpreted as the distribution of the value of the edge incident to $i$ in the maximum matching. This is true because we select a candidate with the probability that it belongs to the maximum matching. Therefore:
	\begin{align}
		\E[\OPT]=\E_{\bval\sim D}[M(\bval)]=\sum_{i=1}^{n}\E_{\hat{\bval}_i\sim \hat{D}_i}[\sum_{k=1}^{m}\hat{v}_{i,k}]\label{eq:matching} \enspace .
	\end{align}

	Now we reduce the problem to the single item case. By looking at a scenario of the problem from the viewpoint of item $k$ we notice that the whole scenario and the algorithm run equivalent to the single item case. This item observes the buyers in a random order such that the valuation of buyer $i$ comes from $\hat{D}_i$. These scenarios occur in parallel for all the items, because no two items are offered to buyers at the same time. In addition, every item $k$ is offered to a buyer with an overall probability of 
	\begin{align*}
		\sum_{i=1}^{n}Pr_{\hat{\bval}_i\sim\hat{D}_i}[\hat{v}_{i,k}>0] &=\sum_{i=1}^{n}\E_{\bval_i, \bval_{-i}}[\one_{v_{i,k}>0} \cdot \one_{(i,k)\in M(\bval_i\cup \bval_{-i})}]\\ &=Pr_{\bval\sim D}[\text{$k$ is matched}]\leq 1\enspace .
	\end{align*}
	Now we can use the result of Claim \ref{claim:stasum}. The right hand side of Equality \eqref{eq:matching} can be written as $\sum_{k=1}^{m}\sum_{i=1}^{n}\E_{\hat{\bval}_i\sim\hat{D}_i}[\hat{v}_{i,k}]$. The claim states that there exists a threshold $\tau_k$ such that the FTA achieves at least $(1-1/e)\sum_{i=1}^{n}\E_{\hat{\bval}_i\sim\hat{D}_i}[\hat{\bval}_{i,k}]$ for every item $k$. Therefore our algorithm is $(1-1/e)$-approximation for the matching of all items.
\end{proof}


\end{document}